\documentclass{llncs}
\usepackage[dvips]{graphicx}
\usepackage{latexsym,tabularx}
\usepackage{amsmath,amssymb}

 \newtheorem{fact}{{\em fact}}

 \newcommand{\logstar}{\log^*\hspace{-.9mm}}

\begin{document}
\title{
A Searchable Compressed Edit-Sensitive Parsing
}

\author{
Naoya Kishiue\inst{1}
Masaya Nakahara\inst{1}
Shirou Maruyama\inst{2}, and\\
Hiroshi Sakamoto\inst{1,3}
}

\institute{
Kyushu Institute of Technology, 680-4 Kawazu, Iizuka-shi, Fukuoka, 820-8502, 
\and Kyushu University, 744 Motooka, Nishi-ku, Fukuoka-shi, Fukuoka 819-0395,
\and PRESTO JST, 4-1-8 Honcho Kawaguchi, Saitama 332-0012, JAPAN\\
\email{
\{n\_kishiue,m\_nakahara\}@donald.ai.kyutech.ac.jp,
shiro.maruyama@i.kyushu-u.ac.jp, 
hiroshi@ai.kyutech.ac.jp,}
}

\date{\empty}
\maketitle

\begin{abstract}
Practical data structures for the edit-sensitive parsing (ESP) are proposed.
Given a string $S$, its ESP tree is equivalent to a context-free 
grammar $G$ generating just $S$, which is represented by a DAG.
Using the succinct data structures for trees and permutations,
$G$ is decomposed to two LOUDS bit strings and single array
in $(1+\varepsilon)n\log n+4n+o(n)$ bits for any 
$0<\varepsilon <1$ and the number $n$ of variables in $G$.
The time to count occurrences of $P$ in $S$ is in
$O(\frac{1}{\varepsilon}(m\log n+occ_c(\log m\log u))$, whereas $m = |P|$, $u = |S|$,
and $occ_c$ is the number of occurrences of a maximal common subtree 
in ESPs of $P$ and $S$.
The efficiency of the proposed index is evaluated by 
the experiments conducted on several benchmarks complying 
with the other compressed indexes.
\end{abstract}



\section{Introduction}
The edit distance is one of the most fundamental problems 
with respect to every string in dealing with the text.
Exclusively with the several variants of this problem, 
the {\em edit distance with move} where moving operation 
for any substring with unit cost is permitted is NP-hard
and $O(\log u)$-approximable~\cite{Shapira07} for string length $u$.
With regard to the matching problem whose approximate solution 
can be obtained by means of {\em edit-sensitive parsing} 
(ESP) technique~\cite{Cormode07}, 
utilization of detected maximal common substrings 
makes it possible to expect application of the problem 
to plagiarism detection and clustering of texts.
As a matter of fact, a compression algorithm based on ESP 
has been proposed~\cite{Sakamoto09}, which results 
in exhibition of its approximation ratio for the optimum compression.

In this work, we propose a practical compressed index for ESP.
Utilization of a compressed index makes it possible to 
search patterns rapidly, which is regarded as a specific case of 
maximum common substrings of the two strings where one is 
entirely in the other.
Comparison of the compressed index proposed in this work with 
the indexes dealt with in the other methods reveals that 
sufficient performance is provided in accordance with the proposed method.
On the other hand, it is shown from theoretical analysis of ESP
that thanks to the proposed method, 
a long enough common substring of the two strings of the text 
and pattern can be found rapidly from the compressed index.

Edit distance problem is closely related to optimum compression.
Particularly with one of the approximation algorithms, 
assigning a same variable to common subtrees allows
approximately optimum parsing tree, 
i.e. approximately optimum CFG to be computed.
This optimization problem is not only NP-hard but also 
$O(\log n)$-approximable~\cite{Charikar05,Rytter03,Sakamoto05}.
As a consequence, compressing two strings and finding out 
occurrences of a maximal subtree from these parsing trees make 
it possible to determine with great rapidity whether one string 
manifests itself in another in a style of a substring.

Our contributions are hereunder described.
The proposed algorithm for indexed grammar-based compression 
outputs a CFG in Chomsky normal form.
The said CFG, which is equivalent to a DAG G where 
every internal node has its left and right children,
is also equivalent to the two spanning trees.
The one called the left tree is exclusively constructed by the left edges, 
whereas the one called the right tree is exclusively 
constructed by the right edges.
Both the left and the right trees are encoded by LOUDS~\cite{Delpratt06}, 
one of the types of the succinct data structure for ordered trees. 
Furthermore the correspondence among the nodes of the trees 
is memorized in an array.
Adding the data structure for the permutation~\cite{Munro03} 
over the array makes 
it possible to traverse the $G$. 
Meanwhile it is possible for the size of the data structure 
to be constructed with $(1+\varepsilon)n\log n+4n+o(n)$ bits
for arbitrary $0<\varepsilon <1$, 
where $n$ is the number of the variables in the $G$.

At the next stage, the algorithm should refer to a 
function, called {\em reverse dictionary}
for the text when compression of the pattern is executed.
For example, if a production rule $Z\to XY$ is included in $G$, 
an occurrence of the digram $XY$ in a pattern,
which is determined to be replaced, 
should be replaced without fail by the same $Z$.
Taking up the hash function $H(XY)=Z$ for the said purpose compels 
the size of the index to be increased. 
Thus we propose the improvement for compression so as to 
obtain the name $Z$ directly from the compression.
It is possible to calculate the number of occurrences of 
a given pattern $P$ from a text $S$ in 
$O(\frac{1}{\varepsilon}(m\log n+occ_c(\log m\log u))$ time
in accordance with the contrivance referred to above together 
with the characteristics of the ESP, where $m = |P|$ and $u = |S|$.
On the other hand, $occ_c$ is the occurrence number of maximal 
common subtree called a core in the parsing tree for $S$ and $P$.
The core is obtained from ESP for $S$ and $P$, 
and it is understood that a constant $\alpha$ is in existence 
to show the lower bound that a core encodes a substring longer 
than $\alpha m$.

At the final stage, comparison is made between the performance 
of our method and that of the other practical 
compressed indexes~\cite{Navarro04JDA,Navarro07,Sadakane03}, called
Compressed Suffix Array (and RLCSA, improved to repetitive texts), 
FM-index, and LZ-index.
Compressed indexes to comply with 200MB English texts, 
DNA sequences, and other repetitive texts are constructed. 
Thereafter comparison is made with the search time to count 
occurrences of patterns to correspond to the pattern length.
As a result, it is ascertained that the proposed index is
efficient enough among these benchmarks 
in case the pattern is long enough to accomplish 
the construction of the indexes.

\section{Preliminaries}
The set of all strings over an alphabet $\Sigma$ 
is denoted by $\Sigma^*$.
The length of a string $w \in \Sigma^*$ is denoted by $|w|$.
A string $\{a\}^*$ of length at least two is called 
a {\em repetition of \/} $a$.
$S[i]$ and $S[i,j]$ denote the $i$-th symbol of $S$ and
the substring from $S[i]$ to $S[j]$, respectively.
The expression $\logstar n$ indicates the maximum number
of logarithms satisfying $\log\log\cdots\log n\geq 1$.
For instance, $\logstar n=5$ for $n=2^{65536}$.
We thus treat $\logstar n$ as a constant.

We assume that any context-free grammar $G$ is 
{\em adimissible}, i.e., $G$ derives just one string. 
For a production rule $X\to AB\cdots C$, symbol $X$ 
is called {\em variable}.
If $G$ derives a string $w$, the derivation is represented by
a rooted ordered tree, called the {\em parsing tree} of $G$.
The {\em size of $G$\/} is the total length of strings in the right
hand sides of all production rules, and is denoted by $|G|$. 
The optimization for the {\em grammar-based compression} is 
to minimize the size of $G$ deriving a given string $w$.
For the approximation ratio of this problem, 
see \cite{Charikar05,Rytter03,Sakamoto05,Sakamoto09}.

We consider a special parsing tree of CFG constructed by 
{\em edit sensitive parsing\/} by~\cite{Cormode07},
which is based on a transformation of string called 
{\em alphabet reduction\/}.
A string $S\in\Sigma^*$ of length $n$ is partitioned into 
maximal nonoverlapping substrings of three types;
Type1 is a maximal repetition of a symbol,
Type2 is a maximal substring longer than $\logstar n$
not containing any repetition, and
Type3 is any other short substring.
Each such substring is called a {\em metablock\/}.
We focus on only Type2 metablocks since the others are not 
related to the alphabet reduction.
From a Type2 string $S$, a label string $label(S)$ is computed as follows.

{\bf Alphabet reduction\/}:
Consider $S[i]$ and $S[i-1]$ represented as binary integers.
Denote by $\ell$ the least bit position in which $S[i]$ differs from $S[i-1]$.
For instance, if $S[i]=101,S[i-1]=100$ then
$\ell=0$, and if $S[i]=001,S[i-1]=101$ then $\ell=2$.
Let $bit(\ell,S[i])$ be the value of $S[i]$ at $\ell$.
Then $label(S[i])=2\ell + bit(\ell,S[i])$.
By this, a string $label(S)$ is obtained as the sequence of such $label(S[i])$.

For the resulting $label(S)$, $label(S[i])\neq label(S[i+1])$ if $S[i]\neq S[i+1]$ for any $i$ 
(See the proof by~\cite{Cormode07}).
Thus the alphabet reduction is recursively applicable to $label(S)$, which is also Type2.
If the alphabet size in $s$ is $\sigma$, the new alphabet size in
$label(S)$ is $2\log\sigma$. 
We iterate this process for the resulting string $label(S)$
until the size of the alphabet no longer shrinks.
This takes $\logstar\sigma$ iterations.

After the final iteration of alphabet reduction, 
the alphabet size is reduced to at most $6$ like $\{0,\cdots,5\}$.
Finally we transform $label(S)\in\{0,\cdots,5\}^*$ 
to the same length string in $label(S)\in\{0,1,2\}^*$ 
by replacing each $3$ with the least integer in $\{0,1,2\}$
that does not neighbor the $3$, and
doing the same replacement for each $4$ and $5$.
We note that the final string $label(S)$ is also Type2 string.
This process is illustrated for a concrete string $S$ in Fig.~\ref{reduction}.

{\bf Landmark\/}: For a final string $label(S)$,
we pick out special locations called landmarks
that are sufficiently close together.
We select any position $i$ as a landmark if $label(S[i])$ is maximal,
i.e., $label(S[i])>label(S[i-1]),label(S[i+1])$.
Following this, we select any position $j$ as a landmark
if $label(S[j])$ is minimal and both $j-1,j+1$ are not selected yet. 
We also display this selection of landmarks in Fig.~\ref{reduction}.

\begin{figure}[bt]
\begin{center}
\includegraphics[scale=.45]{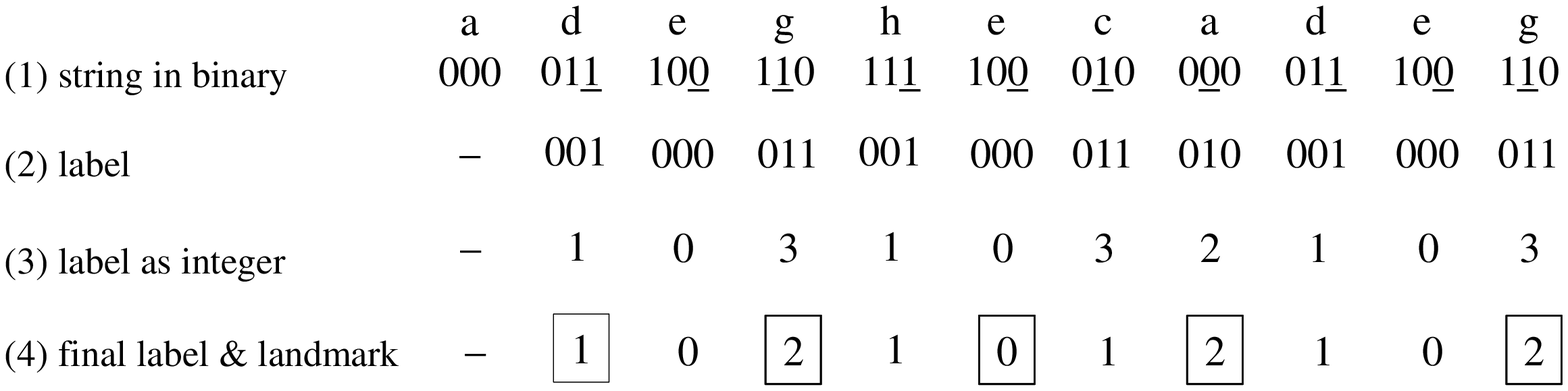}
\end{center}
\caption{
Alphabet reduction:
The line (1) is an original Type2 string $S$ from the alphabet
$\{a,b,\cdots,h\}$ with its binary representation.
An underline denotes the least different bit position to the left. 
(2) is the sequence of $label(S[i])$ formed from
the alphabet $\{0,1,2,3\}$ whose size is less than $6$, and
(3) is its integer representation.
(4) is the sequence of the final labels reduced to $\{0,1,2\}$
and the landmarks indicated by squares.
}
\label{reduction}
\end{figure}

\begin{figure}[tb]
\begin{center}
\includegraphics[scale=.45]{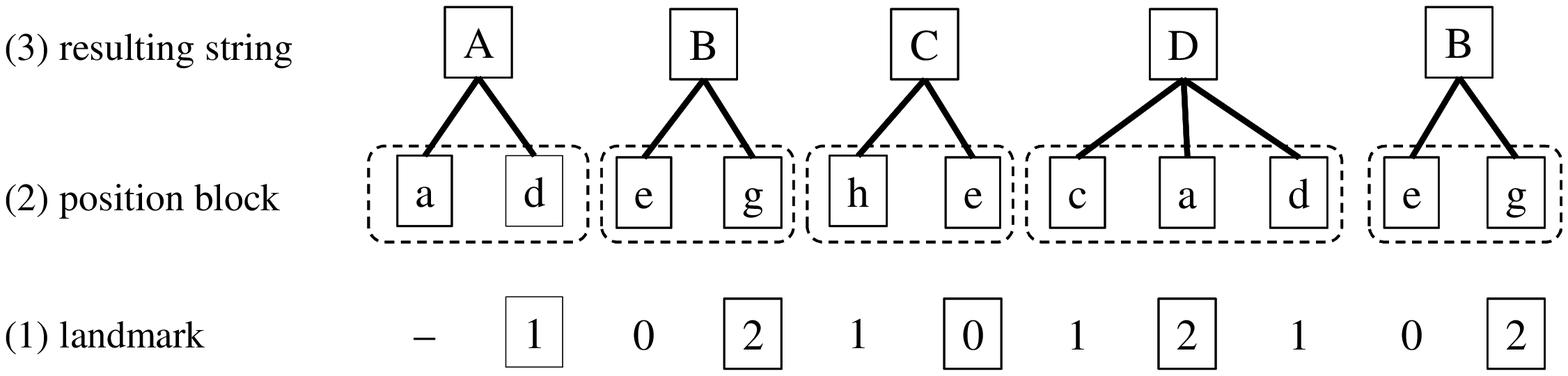}
\end{center}
\caption{
Single iteration of ESP:
The line (1) is the computed final labels and landmarks.
(2) shows the groups of all positions in $s$
having two or three around the landmarks.
(3) is the resulting string $ABCDB$, and
the production rules $A\to ad$, $B\to eg$, etc.
}
\label{esp}
\end{figure}

{\bf Edit sensitive parsing\/}:
After computing final string $label(S)$ and its landmarks for a Type2 string $S$, 
we next partition $S$ into blocks of length two or three around
the landmarks in the manner:
We make each position part of the block generated by its
closest landmark, breaking ties to the right.

Since $label(S)\in\{0,1,2\}^*$ contains no repetition,
for any two successive landmark positions $i$ and $j$,
we have $2\leq |i-j|\leq 3$.
Thus, each position block is of length two or three.
The string $S$ is transformed to a shorter string $S'$ 
by replacing any block of two or three symbols to a new suitable symbol.
Here ``suitable'' means that any two blocks for a same substring
must be replaced by a same symbol.
This replacement is called {\em edit sensitive parsing\/} (ESP).
We illustrate single iteration of ESP for determined blocks in Fig.~\ref{esp}.

Finally, we mention Type1 or Type3 string $S$.
If $|S|\geq 2$, we parse the leftmost two symbols of $S$
as a block and iterate on the remainder and
if the length of it is three, then we parse the three symbols as a block. 
We note that no Type1 $S$ in length one exists.
The remaining case is Type3 $S$ and $|S|=1$, which appears in a context $a^*bc^*$.
If $|a^*|=2$, $b$ is parsed as the block $aab$.
If $|a^*|>2$, $b$ is parsed as the block $ab$.
If $|a^*|=0$, $b$ is parsed with $c^*$ analogously.

If $S$ is partitioned into $S_1,\ldots,S_k$ of Type1, Type2, or Type3,
after parsing them, all the transformed strings $S'_i$ are concatenated together.
This process is iterated until a tree for $S$ is constructed.
By the parsing manner, we can obtain a balanced $2-3$ tree,
called {\em ESP tree}, in which any internal node has two or three children.

\section{Algorithms and Data Structures}
In this section, it is shown that searching a pattern in a text 
is reduced to finding some adjacent subtrees in the ESP trees 
corresponding to the pattern and text. 
This problem is solved by practical algorithms and data structures.

\subsection{Basic notions}
A set of production rules of a CFG is represented by 
a directed acyclic graph (DAG) with the root labeled by the start symbol.
In Chomsky normal form hereby taken up, each internal node has 
respectively two children called the left/right child, 
and each edge is also called the left/right edge.
An internal node labeled by $X$ with 
left/right child labeled by $A$/$B$ is corresponding to 
the production rule $X\to AB$.
We note that this correspondence is one-to-one so that the DAG of a CFG 
$G$ is a compact representation of the parsing tree $T$ of $G$.
Let $v$ be a node in $T$, and the subtree of $v$ is the induced 
subgraph by all descendant of $v$.
The parent, left/right child, and variable on a node $v$ is 
denoted by $parent(v)$, $\mbox{\em left}(v)/\mbox{\em right}(v)$, 
and $label(v)$, respectively.

A {\em spanning tree\/} of a graph $G$ is a subgraph of $G$
which is a tree containing all nodes of $G$.
A spanning tree of a DAG is called {\em in-branching} 
provided that the out-degree of each node except the root 
is a single entity, and the {\em out-branching} spanning tree 
is the reverse notion.

With respect to an ordered binary tree $T$, a node $v$ is called 
the {\em lowest right ancestor} of a node $x$ and is denoted by $lra(x)$,
provided that $v$ is the lowest ancestor so that the path from $v$ 
to $x$ will contain at least one left edge.
If $x$ is a node in the right most path in $T$, $lra(x)$ is undefined.
Otherwise, $lra(x)$ is uniquely decided.
The subtree of $x$ is {\em left adjacent} to the subtree of $y$
provided that $lra(x) =lla(y)$,
thus the {\em adjacency in the right} is similarly defined.
These notions are illustrated in Fig.~\ref{adjacent}, 
from which the characterization shown below can be obtained.

\begin{fact}\rm\label{lra}
For an ordered binary tree,
a node $y$ is right adjacent to a node $x$ 
iff $y$ is in the left most path from $\mbox{\em right}(lra(x))$,
and $y$ is left adjacent to $x$ iff
$y$ is in the right most path from $\mbox{\em left}(lla(x))$.
\end{fact}

Checking such adjacency is a basic operation of the proposed algorithm
to decide the existence of patterns from the compressed string. 
The efficiency is guaranteed by several techniques 
introduced in the following subsections.

\begin{figure}[bt]
\begin{center}
\includegraphics[scale=.50]{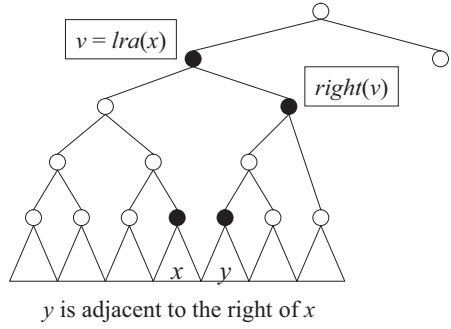}
\end{center}
\caption{
The relation of two nodes $x$ and $y$ in a rooted ordered binary tree.
They are adjacent in this order iff
$y$ is in the left most path from $\mbox{\em right}(lra(x))$ as illustrated.
}
\label{adjacent}
\end{figure}

\subsection{Pattern embedding on parsing tree}
For two parsing trees of strings $P$ and $S$,
if there is a common subtree for them,
then its root variable is called a {\em core}.
It is shown that with respect to each of strings $P$ and $S$, 
these ESP trees concerning a same naming function contain
a sufficiently large core $X$ provided $S$ contains $P$.
This property is available as a necessary condition in searching $P$. 
In other words, any occurrence of $P$ in $S$ is restricted 
in a region around $X$.

\begin{lemma}\label{core-lemma}\rm
There exists a constant $0<\alpha<1$
such that for any occurrence of $P$ in $S$,
its core is encoding a substring longer than $\alpha |P|$. 
\end{lemma}
\begin{proof}
We first consider the case that $P$ is a Type2 metablock.
As shown by~\cite{Cormode07}, determining the closest landmark on $S[i]$
depends on $S[i-\logstar n+5,i]$ and $S[i,i+5]$.
Thus, if $S[i,j]=P$, then
the final labels for the inside part $S[i+\logstar n+5,j-5]$
are the same for any occurrence position $i$ of $P$.
The above mentioned matter allows each substring equivalent to 
$S[i+\logstar n+5,j-5]$ to be transferred to a same $S'$.
Since the ESP tree is balanced $2-3$ tree,
any variable in $S'$ encodes at least two symbols.
If $S'$ assumes Type2 again, then this process iterated.
Thus, after $k$ iterations, the length of the string 
encoded by a variable in $S'$ is at least $2^k$.
Meanwhile owing to one iteration, the common substring $S'$ 
loses its prefix and suffix of length at most $\logstar n+5$.
In addition, each lost variable has no less than three children.
By the above observation, we can take an internal node
as a core of $P$ for $S$, whose height is the maximum $k$ satisfying
\[
2(\logstar n+5)(3+3^2+\cdots 3^k)
<(\logstar n+5)3^{k+2}
\leq |P|.
\]
In consideration of the above estimation together with 
the fact that $\logstar n$ is regarded as a constant and 
concurrently a variable in height $k$ encodes a substring of 
the length of the minimum $2^k$, 
a constant $0<\alpha<1$ and a variable is obtained 
as a core of $P$ encoding a substring of length at least $\alpha|P|$.
$P$ is generally divided into metablocks as seen 
in a manner of $P=P_1P_2\cdots P_m$.
Type1 and Type3 metablocks in $P_2\cdots P_{m-1}$ are uniquely parsed
in its any occurrence.  
Thus we can assume $P=P_1P_2P_3$ for a long Type2 metablock $P_2$ 
and Type1 $P_1,P_3$ as a worst case.
For any occurrence of Type1 metablock, 
we can obtain a sufficiently large core.
Choosing a largest core from the three metablocks,
the size is greater than $\alpha |P|$.
\end{proof}

Using Lemma~\ref{core-lemma}, the search problem for $P$ is 
reduced to the other problem for the sequence of adjacent cores.

\begin{lemma}\label{induce-lemma}\rm
For a given ESP tree $T$ of a text $S$ and a pattern $P$, $S[i,j]=P$ iff 
there exist $k=O(\log|P|)$ adjacent subtrees in $T$  
rooted by variables $X_1,\ldots,X_k$ such that the concatenation of all strings
encoded by them is equal to $P$.
\end{lemma}
\begin{proof}
If the bound $k=O(\log|P|)$ is unnecessary, 
trivial subtrees equal to the leaves 
$S[i],S[i+1],\ldots,S[j]$ can always be obtained.
Use of Lemma~\ref{core-lemma} makes it possible to find 
a core that encodes a long substring of $S[i,j]$
longer than $\alpha|j-i|$ for a fixed $0<\alpha<1$.
The remaining substrings are also covered by their own cores, 
from which the bound $k=O(\log |P|)$ is obtained.
\end{proof}

Two algorithms are developed for compression and search based 
on Lemma~\ref{core-lemma} and~\ref{induce-lemma}.
At first, since any ESP tree is balanced $2-3$ tree, 
each production rule is of $X\to AB$ or $X\to ABC$.
The latter is identical to $X\to AB'$ and $B'\to BC$.
Assumption is hereby made exclusively with Chomsky normal form.
A data structure $D$ to access the digram $XY$ from a variable $Z$ 
associated by $Z\to XY$ is called a {\em dictionary}. 
In the meantime, another data structure $D^R$ to compute the reverse 
function $f(XY)=Z$ is called a {\em reverse dictionary}.

{\em ESP-COMP} is described in Fig.~\ref{ESP-COMP} 
with a view to computing the ESP tree of a given string.
This algorithm outputs the corresponding dictionary $D$.
The reverse dictionary $D^R$ is required to replace different 
occurrences of $XY$ by means of a common variable $Z$.
This function, which can be developed by a hash function 
with high probability~\cite{Karp87}, requires large extra 
space regardless of such a circumstance.
In the next subsection, we propose a method to simulate $D^R$ by $D$.
The improvement brought about as above makes it possible to 
compress a given pattern for the purpose of obtaining 
the core exclusively by $D$.

{\em ESP-SEARCH} is described in Fig.~\ref{ESP-SEARCH} to 
count occurrences of a given pattern $P$ in $S$.
To extract the sequence of cores, $P$ is also compressed by {\em ESP-COMP} 
referring to $D^R$ for $S$. Furthermore if $XY$ is undefined in $D^R$, 
a new variable is produced and $D^R$ is updated.
Then {\em ESP-SEARCH} gets the sequence of cores, $X_1,\ldots,X_k$
to be embedded on the parsing tree of $S$.
The algorithm  checks if $X_i$ is left adjacent to $X_{i+1}$ 
for all $i=1,\ldots,k-1$ from a node $v$ labeled by $X_1$.
As we propose several data structures in the next subsection, 
we can access to all such $v$ randomly.
Thus, the computation time is faster than the time to 
traverse of the whole ESP tree, which is proved by the time complexity.

\begin{figure}[tb]
\noindent
\hrulefill
\begin{tabbing} 
~~~~ \= ~~~ \= ~~~ \= ~~~ \= ~~ \= ~~ \= ~~ \= \kill 
\> {\em Algorithm} {\bf\em ESP-COMP} \\
\>Input: a string $S$.\\
\>Output: a CFG represented by $D$ deriving $S$.\\
\>\>initialize $D$;\\
\>\>{\bf while}($|S|>1$)\\
\>\>\>{\bf for-each}($X_k\to X_iX_j$ produced in same level of ESP)\\
\>\>\>\>sort all $X_k\to X_iX_j$ by $(i,j)$;\\
\>\>\>\>rename all $X_k$ in $S$ by $X_\ell$, 
the rank of sorted $X_k\to X_iX_j$;\\ 
\>\>\>\>update $D$ for renovated $X_\ell\to X_iX_j$;\\
\>\>return $D$;\\\\
\> {\em procedure\/} $ESP(S,D)$\\
\>\>compute one iteration of ESP for $S$;\\
\>\>update $D$;\\
\>\>return the resulting string;
\end{tabbing}
\noindent
\hrulefill
\caption{The compression algorithm 
to output a dictionary $D$ for a string $S$. 
We assume the reverse dictionary $D^R$.
}
\label{ESP-COMP}
\end{figure}

\begin{figure}[tb]
\noindent
\hrulefill
\begin{tabbing} 
~~~~ \= ~~~ \= ~~~ \= ~~~ \= ~~ \= ~~ \= ~~ \= \kill 
\> {\em Algorithm} {\bf\em ESP-SEARCH} \\
\>Preprocess: $D\leftarrow ${\em ESP-COMP}$(S)$ for text $S$.\\
\>Input: a pattern $P$.\\
\>Output: the number of occurrences of $P$ in $S$\\
\>\>$count\leftarrow 0$ and $(X_1,\ldots,X_k)\leftarrow FACT(P,D)$;\\
\>\>{\bf for-each}($v$ satisfying $label(v)=X_1$)\\
\>\>\>$i\leftarrow 2$, $t\leftarrow \mbox{\em right}(lra(v))$, and
$type \leftarrow \mbox{true}$;\\
\>\>\>{\bf while}($i\leq k$)\\
\>\>\>\>{\bf if}(a left descendant $v'$ of $t$ satisfies
$label(v')=X_i$)\\
\>\>\>\>\>\>$v\leftarrow v'$, $t\leftarrow \mbox{\em right}(lra(v))$,
and $i\leftarrow i+1$;\\ 
\>\>\>\>{\bf else} $type \leftarrow \mbox{false}$, and break; \\
\>\>\>{\bf if}($type = \mbox{true}$), $count\leftarrow count+1$; \\
\>\> return $count$; \\\\
\> {\em procedure} $FACT(P,D)$\\
\>\>compute the variable by $CORE(P,D)$ which encodes $P[i,j]$;\\
\>\>recursively compute the variables\\
\>\>\>$CORE(pre(P),D)$ for $pre(P)=P[1,i-1]$ and\\
\>\>\>$CORE(suf(P),D)$ for $suf(P)=P[i+1,|P|]$;\\
\>\>return all variables from the left occurrence;\\\\
\> {\em procedure\/} $CORE(P,D)$\\
\>\>$\ell\leftarrow 1$ and $r\leftarrow |P|=m$;\\
\>\>{\bf while}($|P|>1$ and $\ell < r$)\\
\>\>\>$P\leftarrow ESP(P,D)$\\
\>\>\>$\ell\leftarrow (\ell + \lceil \logstar n\rceil+5$) and $r\leftarrow r-5$;\\
\>\>return the symbol $P[1]$;
\end{tabbing}
\noindent
\hrulefill
\caption{The pattern search algorithm 
from the compressed text represented by a dictionary $D$.
We assume the reverse dictionary $D^R$ again.
}
\label{ESP-SEARCH}
\end{figure}

\begin{lemma}\rm\label{th-algo1}
If we assume the reverse dictionary $D^R$ with constant time access,
the running time of {\em ESP-COMP} is $O(u)$
and the height of the ESP tree is $O(\log u)$
for the length of string, $u$.
\end{lemma}
\begin{proof}
The algorithm shortens a current string to at least half by each
iteration, and all the digrams are sorted in linear time 
by the radix sort in the procedure. 
This outer loop is executed $O(\log u)$ times. 
Thus, the bound is obtained.
\end{proof}

In {\em ESP-SEARCH}, several data structures are assumed
and they are developed in the next subsection.
At this stage the correctness is exclusively ensured, 
which is derived from Lemma~\ref{core-lemma} and~\ref{induce-lemma}.

\begin{lemma}\rm\label{th-algo2}
{\em ESP-SEARCH} correctly counts the occurrences
of a given pattern in the ESP tree of a text.
\end{lemma}

The time/space complexity of the algorithms depends on 
the performance of the data structure employed.
As a matter of fact, the size of the parsing tree is greater than 
the length of the string for a naive implementation. 
In the next subsection, proposal is made with a compact representation 
of parsing tree and reverse dictionary for the algorithm.

\subsection{Compact representation for ESP}

We propose compact data structures used by the algorithms.
These types of improvement are achieved by means of two techniques: 
one is the decomposition of DAG representation into left/right tree, 
and the other is the simulation of the reverse dictionary $D^R$ 
by the dictionary $D$ with an auxiliary data structure.
First the decomposition of DAG is considered.
Let $G$ be a DAG representation of a CFG in Chomsky normal form.
By introducing a node $v$ together with addition of left/right 
edges from any sink of $G$ to $v$, $G$ can be modified to 
have the unique source and sink.

\begin{fact}\rm\label{spanning-lemma}
Let $G$ be a DAG representation with single source/sink
of a CFG in Chomsky normal form.
For any in-branching spanning tree of $G$,
the graph defined by the remaining edges is also an 
in-branching spanning tree of $G$.
\end{fact}

An in-branching spanning tree of $G$, which is called the {\em left tree} 
of $G$, is concurrently denoted $T_L$ provided that the tree consists 
exclusively of the left edges. 
Thus the complementary tree is called the {\em right tree} of $G$ 
to be denoted $T_R$.
A schematic of such trees is given in Fig.~\ref{LRtree}.

When a DAG is decomposed into $T_L$ and $T_R$, 
the two are represented by succinct data structures 
for ordered trees and permutations.
Brief description concerning the structures is hereunder made. 
The bit-string by LOUDS~\cite{Delpratt06} for an ordered tree 
is defined as shown below.
We visit any node in level-order from the root.
As we visit a node $v$ with $d\geq 0$ children, we append $1^d0$
to the bit-string beginning with the empty string.
Finally, we add $10$ as the prefix corresponding to an
imaginary root, which is the parent of the root of the tree.
A schematic of the LOUDS representations for $T_L$ and $T_R$ 
is also given in Fig.~\ref{LRtree}.
For $n$ node tree, LOUDS uses $2n+o(n)$ bits 
to support the constant time access to the parent, 
the $i$-th child, and the number of children of a node,
which are required by our algorithm.

For traversing the DAG, we also need the correspondence of 
the set of nodes in one tree to the one in the other.
For this purpose,
we employ the succinct data structure for permutations by~\cite{Munro03}.
For a given permutation $P$ of $N=(0,\ldots,n-1)$,
using $(1+\varepsilon)n\log n+o(1)$ bits space,
the data structure supports to access to $P[i]$ in $O(1)$ time
and $P^{-1}[i]$ in $O(1/\varepsilon)$ time.
For instance, if $P=(2,3,0,4,1)$, then $P[2]=0$ and $P^{-1}[4]=3$,
that is, $P[i]$ is the $i$-th member of $P$ and
$P^{-1}[i]$ is the position of the member $i$.
For each node $i$ in $LOUDS(T_L)$, the corresponding node $j$ 
in $LOUDS(T_R)$ is stored in $P[i]$.
These are also illustrated in Fig.~\ref{LRtree}.

In the compression algorithm in Fig.~\ref{ESP-COMP},
all variables produced in a same level
are sorted by the left hands of production rules\footnote{
In~\cite{Claude09}, 
similar technique was proposed, but variables are sorted by
encoded strings.},
and these variables are renamed by their rank.
Thus, the $i$-th variable in a DAG coincides with node $i$ in $T_L$
since they are both named in level-order.
In accordance with the improvement referred to above, 
storage can be made with the required correspondence 
in nearly $n\log n$ bits. 
Devoid of these characteristics, $2n\log n$ bits 
are required to traverse $G$.

\begin{figure}[bt]
\begin{center}
\includegraphics[scale=.33]{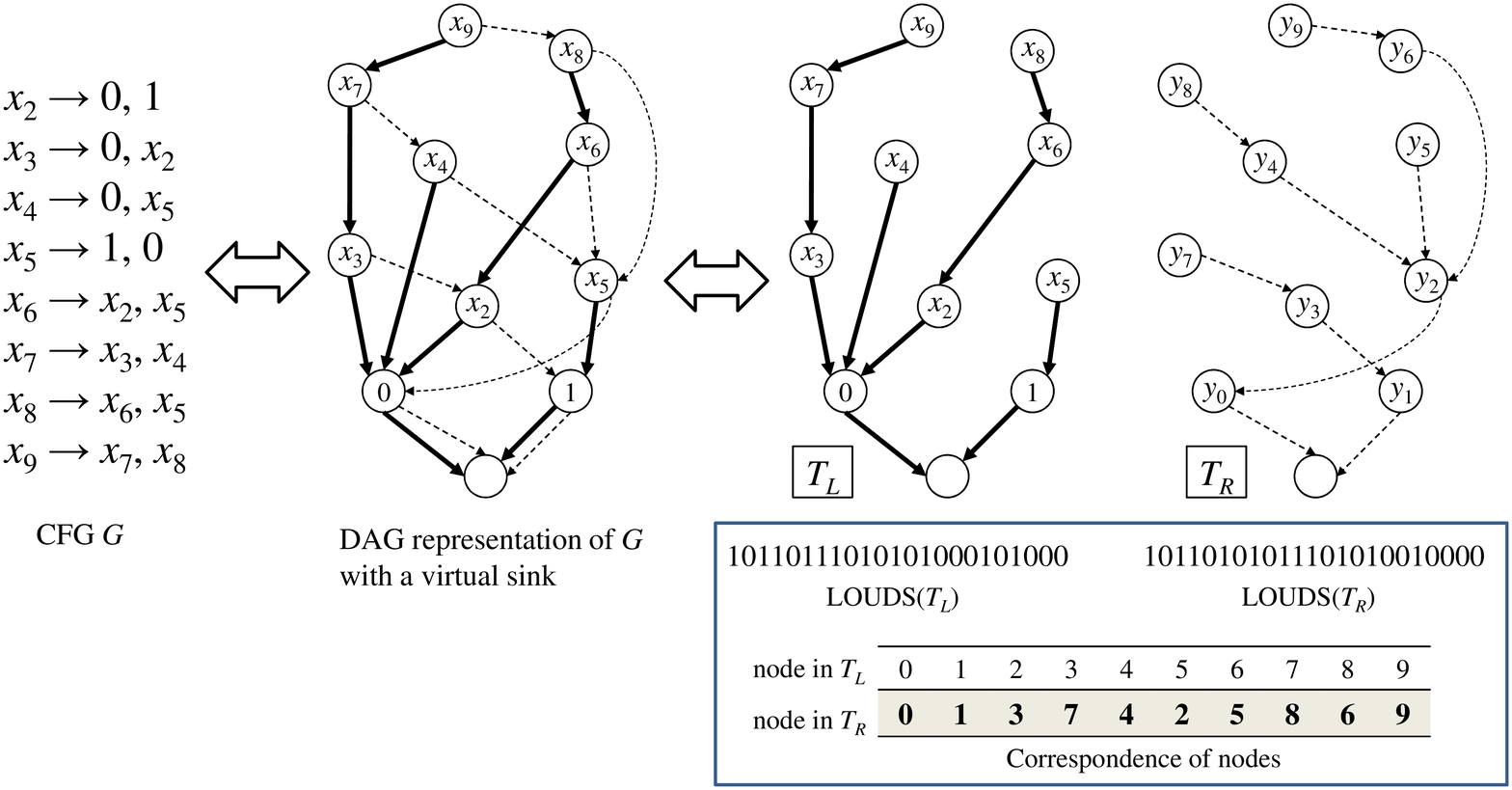}
\end{center}
\caption{
A DAG representing a CFG in Chomsky normal form
and its decomposition into two ordered trees
with their succinct representations.
}
\label{LRtree}
\end{figure}

At the final stage, a method is proposed with a view to 
simulating the reverse dictionary $D^R$ from the data structures 
referred to above.
Adapting this technique makes it possible to reduce the space 
for the hash function to compress a pattern.
Preprocessing causes the $X_k$ to denote the rank of the sorted $X_iX_j$
by $X_k\to X_iX_j$. 
Conversely being given a variable $X_i$, 
the children of $X_i$ in $T_L$ are already sorted by the indexes of 
their parents in $T_R$.
Thus the variable $X_k$ associated to $X_iX_j$
can be obtained by using binary search on the children of $X_i$ in $T_L$, 
of which depiction is made in Fig.~\ref{reverseD}.
Since LOUDS supports the number of the children and $i$-th child, 
access can be made to the middle child $X_i$ in $O(1)$ time. 
Thus we obtain the following lemma. 

\begin{figure}[bt]
\begin{center}
\includegraphics[scale=.4]{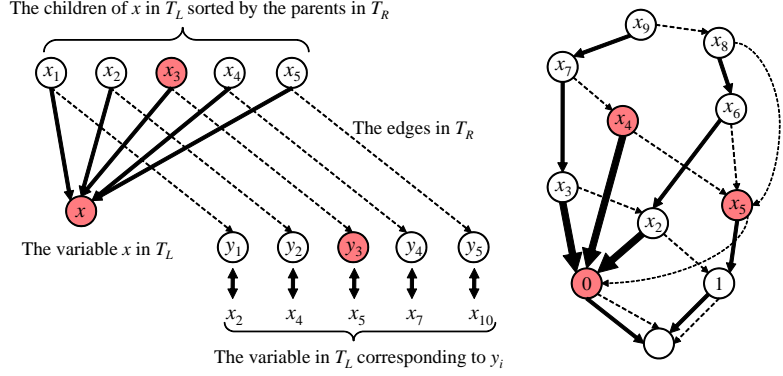}
\end{center}
\caption{The simulation of $D^R$ using binary search 
over the nodes of $T_L$.
For each node $x$ in $T_L$, 
the children $x_i$s of $x$ are already sorted by
the variables in $T_L$
corresponding to the parents of $x_i$s in $T_R$.
}
\label{reverseD}
\end{figure}

\begin{lemma}\rm\label{reverseD-lemma}
The function $f(XY)=Z$ is computable in $O(\frac{1}{\varepsilon}\log k)
=O(\frac{1}{\varepsilon}\log n)$ time for the maximum degree of $T_L$, $k$, 
which is bounded by the number of variables, $n$.
\end{lemma}
\begin{proof}
The statement is derived from the above observation.
\end{proof}

Using the proposed data structures, the
following theorem is obtained.

\begin{theorem}\rm\label{main-th}
A grammar-based compression $G$ for any string $S$ is represented 
in $(1+\varepsilon)n\log n+4n+o(n)$ bits, 
where $n$ is the number of variables in $G$.
With any pattern $P$, the number of its occurrence in $S$ 
is computable in $O(\frac{1}{\varepsilon}(m\log n+occ_c(\log m\log u)))$ time
for any $0<\varepsilon<1$, 
where $u=|S|$, $m=|P|$, and $occ_c$ is the number of occurrences
of a maximal core of $P$ for $S$.
\end{theorem}
\begin{proof}
When the cores $X_1,\ldots,,X_k$ are obtained by the procedure $FACT(P,D)$,
let $X_i$ be one of them.
Modification can easily be made with the search algorithm
to check both the left adjacency of $X_1,\ldots,X_k$ and 
the right adjacency of $X_{i+1},\ldots,X_k$ starting at $X_i$.
Thus the search time is bounded by $occ_c$ choosing 
a maximal core from $X_1,\ldots,X_k$.
\end{proof}

\section{Experiments}
The experiments are conducted in the environment shown below.
OS:CentOS 5.5 (64-bit), CPU:Intel Xeon E5504 2.0GHz (Quad)$\times$2,
Memory:144GB RAM, HDD:140GB, and Compiler:gcc 4.1.2.

Datasets are obtained from the text collection in Pizza\&Chili 
Corpus\footnote{http://pizzachili.dcc.uchile.cl/texts.html} 
to compare hereto referred method called ESP with other compressed 
indexes called  LZ-index (LZI)\footnote{http://pizzachili.dcc.uchile.cl/indexes/LZ-index/LZ-index-1}, Compressed Suffix Array, and FM-index (CSA and FMI)\footnote{http://code.google.com/p/csalib/}.
These implementations are based on~\cite{Navarro04JDA,Navarro07,Sadakane03}.
Due to the trade-off in the construction time and the index size, 
the index referred to above and other methods for reasonable parameters 
are examined.
In our algorithm, setting is made with $\varepsilon=1,1/4$ for the permutation.
In CSA, the option $(\mbox{-P1:L})$ means that
$\psi$ function is encoded by the gamma function
and $\mbox{L}$ specifies the block size for storing $\psi$.
In FMI, $(\mbox{-P4:L})$ means that
BW-text is represented by Huffman-shaped wavelet tree
with compressed bit-vectors
and $\mbox{L}$ specifies the sampling rate for storing rank values,
and $(\mbox{-P7:L})$ is the uncompressed version.
In addition these CSA and FMI do not make indexes for occurrence position.
Setting up is made with 200MB texts for each DAN and ENGLISH 
to evaluate construction time, index size, and search time.

\begin{figure}[tb]
\begin{center}
\includegraphics[scale=.33]{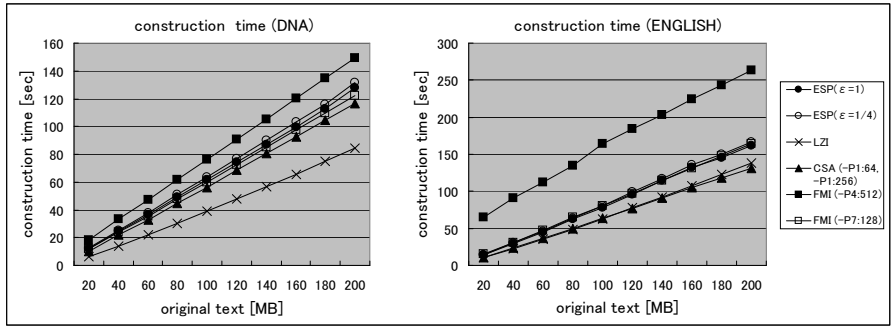}
\end{center}
\caption{Construction Time.}
\label{construction-time}
\end{figure}

The results in construction time are shown in Fig.~\ref{construction-time}.
It is deduced from these results that the method dealt with 
at this stage is comparable with FMI and CSA in the parameters
in construction time, and slower than LZI.
Furthermore it is understood that none of conspicuous difference 
is seen in construction time so long as the value of $\varepsilon$ 
stand still from 1 to $1/4$.

\begin{figure}[tb]
\begin{center}
\includegraphics[scale=.33]{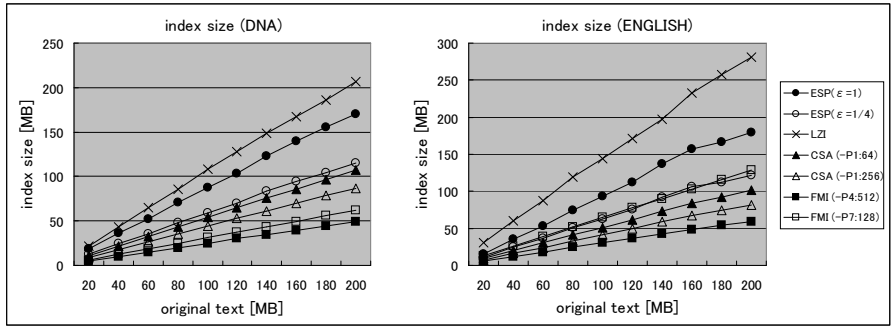}
\end{center}
\caption{Index Size.}
\label{index-size}
\end{figure}

The results of index size are shown in Fig.~\ref{index-size}.
The results reveal that the index is furthermore compact enough 
and comparable to CSA(-P1:64).
The size of LZI contains the space to locate patterns.

\begin{figure}[tb]
\begin{center}
\includegraphics[scale=.33]{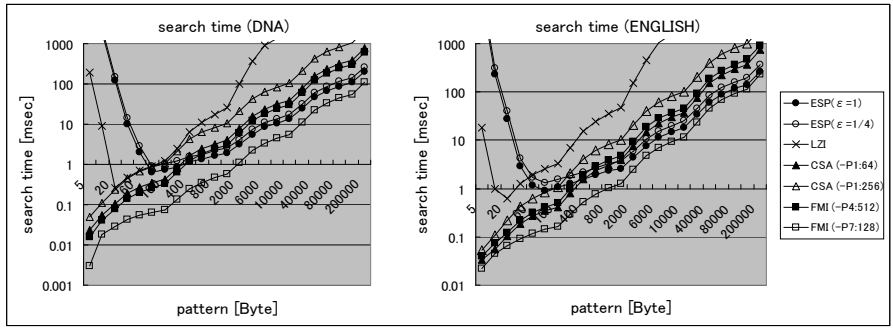}
\end{center}
\caption{Search Time.}
\label{search-time}
\end{figure}

The indexes in Fig.~\ref{search-time} show 
the time to count all occurrences of 
a given pattern in the text.
The indexes are aligned to accomplish the maximum texts in DNA and 
ENGLISH (200MB each).
Random selection of pattern from the text is made 1000 times 
for each fixed pattern length,
and the search time indicates the average time.
In this implementation, we modified our search algorithm 
so that the core is extracted by a short prefix of a given pattern $P$ 
and an occurrence of $P$ in $S$ is decided by the single core and 
the exact match of the remaining substrings by partial decoding of 
the compressed $S$.
To determine length or the short prefix, the rate $1\%$ of 
the pattern by preliminary experiments is taken up.
In DNA and ENGLISH, our method is faster searchable 
than LZI and CSA in the parameters for long patterns.
The proposed method is liable to be behind the pattern 
with short length in case of searching, which might be for
the reason why the occurrence number is relatively made multiplied, 
and comparison of variables are executed for the individual occurrences. 

From the experimental result referred to above, 
it is ascertained that the proposed method, 
which is believed to be subject to settlement of pattern length 
or parameter settlement, can acquire sufficient performance 
as index for pattern searching.

In addition we examine the effect of the parameter $\varepsilon$.
Fig.~\ref{epsilon_trade-off} shows the tradeoff of search time and index size for $\varepsilon$.
The ESP index is constructed for ENGLISH 100MB and the length of pattern is fixed by 10000.
By this figure, the setting $\varepsilon =1/4$ is reasonable.

\begin{figure}[tb]
\begin{center}
\includegraphics[scale=.33]{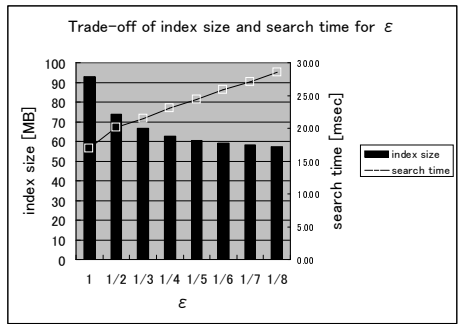}
\end{center}
\caption{Effect of parameter $\varepsilon$.}
\label{epsilon_trade-off}
\end{figure}

We show further experimental results in 
repetitive texts\footnote{http://pizzachili.dcc.uchile.cl/repcorpus.html}
to compare ESP index with another index specifically oriented to repetitive texts, 
called RLCSA\footnote{http://pizzachili.dcc.uchile.cl/indexes/RLCSA/}.
The results are shown in Fig.~\ref{construction-time_r}, Fig.~\ref{index-size_r},
and Fig.~\ref{search-time_r}.
These results reinforce the efficiency of ESP index.

\begin{figure}[tb]
\begin{center}
\includegraphics[scale=.33]{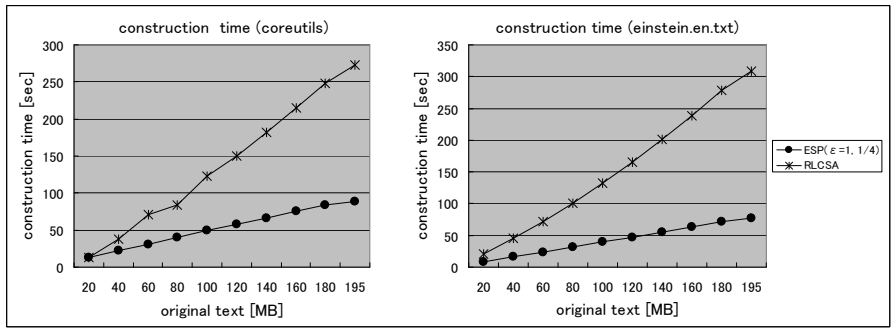}
\end{center}
\caption{Construction time for repetitive texts.}
\label{construction-time_r}
\end{figure}

\begin{figure}[tb]
\begin{center}
\includegraphics[scale=.33]{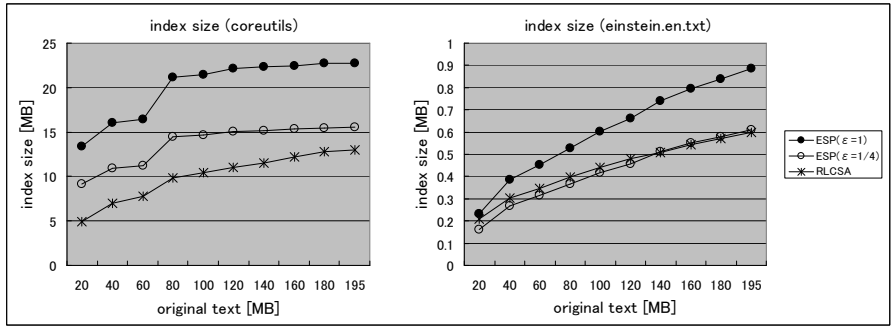}
\end{center}
\caption{Index size for repetitive texts.}
\label{index-size_r}
\end{figure}

\begin{figure}[tb]
\begin{center}
\includegraphics[scale=.33]{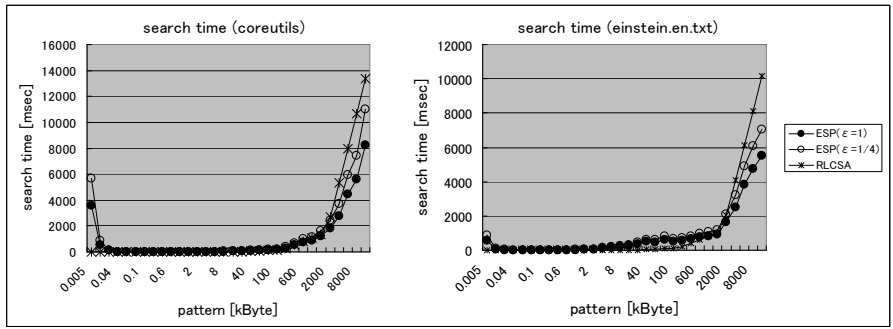}
\end{center}
\caption{Search time for repetitive texts.}
\label{search-time_r}
\end{figure}

\section{Discussion}
We proposed a searchable grammar-based compression for ESP.
Theoretically, this improves the size of naive representation of CFG and 
supports several operations for the compressed strings,
and its performance was confirmed by the implementation
for several benchmarks.

We have another motivation to apply our data structures to practical use.
Originally, ESP was proposed to solve a difficult variant 
of the edit distance problem by finding many maximal common 
substrings of two strings.
Thus, our method will exhibit its ability in case that 
a pattern is as long as a string.
Such situation is found in the framework of
normalized compression distance~\cite{Cilibrasi05}
to compare two long strings directly.
Then we can extract similar parts from very large texts by compression.

\bibliographystyle{plain}
\bibliography{compress,book,cpm,succinct}  

\end{document}